\documentclass[12pt]{article}
\usepackage[utf8]{inputenc}
\usepackage{macros}
\usepackage[margin=1.25in]{geometry}

\title{Robustness of voting mechanisms to external information in expectation}
\author[1,2]{Yiling Chen} 
\affil[1]{Harvard University} 
\author[1, 2]{Jessie Finocchiaro} 
\affil[2]{Harvard Center for Research on Computation and Society}
\affil[ ]{\texttt{\href{mailto:yiling@seas.harvard.edu}{yiling@seas.harvard.edu}}, \texttt{\href{mailto:finocch@bc.edu}{finocch@bc.edu}}}
\date{}

\begin{document}

\maketitle

\begin{abstract}
    Analyses of voting algorithms often overlook informational externalities shaping individual votes.
    For example, pre-polling information often skews voters towards candidates who may not be their top choice, but who they believe would be a worthwhile recipient of their vote.
    In this work, we aim to understand the role of external information in voting outcomes.
    We study this by analyzing (1) the probability that voting outcomes align with external information, and (2) the effect of external information on the total utility across voters, or \emph{social welfare}.  
    In practice, voting mechanisms elicit coarse information about voter utilities, such as ordinal preferences, which initially prevents us from directly analyzing the effect of informational externalities with standard voting mechanisms. 
    To overcome this, we present an intermediary mechanism for learning how preferences change with external information which does not require eliciting full cardinal preferences.
    With this tool in hand, we find that voting mechanisms are generally more likely to select the alternative most favored by the external information, and when external information reflects the population's true preferences, social welfare increases in expectation.
\end{abstract}

\section{Introduction}

Classical results in voting theory prove that no voting mechanism, which is both non-dictatorial and anonymous, can be entirely immune to strategic manipulation~\citep{gibbard1973manipulation,satterthwaite1975strategy}. 
This phenomenon is often evident in real-world scenarios: for instance, a voter who supports candidate $a$ might opt to vote for candidate $b$ instead if candidates $b$ and $c$ are leading in the polls (see, for example, \citep{frenchstrategicvoting}). 
This strategic behavior often appears to be ``directed'' in practice, influencing voters to align their preferences with a piece of external information. 
We explore how this external information, which we call \emph{anchoring information}, (1) affects the probability its most favored alternatives are selected by different voting mechanisms.
When external information is in line with societal preferences, we additionally examine (2) the impact of external information on total utility over the population, measured in social welfare.

Understanding the effects of anchoring information on voting outcomes is not straightforward: most voting mechanisms only elicit \emph{ordinal} preferences from voters, but \emph{cardinal} preferences conveying the strength of preferences are typically needed to know when a voter will respond to anchoring information.
In practice, cardinal preferences are often unknown and impractical to elicit from voters.
However, without access to cardinal utilities, it is unknown how to reason about the effect of anchoring information because we have no information about the strength of preferences within the available ordinal preferences.
Therefore, we must find some way to understand the effect of anchoring information without eliciting cardinal utilities.
Only once we have done this can we begin to study the impact of anchoring information on voting outcomes in expectation.

To this end, we introduce an intermediate mechanism in \S~\ref{sec:elicit-social} that enables us to analyze the effect of anchoring information in order to answer the following three questions:

\begin{enumerate}
    \item In \S~\ref{sec:sw-standard}, we obtain bounds on the probability an alternative is selected by a voting rule; do these bounds tighten when voters have anchored preferences? (Corollaries~\ref{cor:social-prefs-tighten-lb} and \ref{cor:social-prefs-loosen-ub})
    \item When does the expected social welfare increase with anchored preferences? (Theorem~\ref{thm:inc-SW}) 
    \item Can we bound the probability that social welfare decreases on any given voting instance? (Theorem~\ref{thm:bound-prob-dec})
\end{enumerate}

\subsection{Related work}
\paragraph{Anchored preferences}
Recently, \citet{flanigan2023distortion} use implicit utilities (term from~\citep{caragiannis2017subset}) to examine the distortion that emerges under \emph{public-spirited voting}.
Their model of public-spirited utilities is a direct analog of the model presented by \citet{chen2014altruism}, which studies the price of anarchy in congestion games with altruistic players.
Our model of utilities is very similar to the previous models, but is slightly different in that \citet{flanigan2023distortion} and \citet{chen2014altruism}, inspired by \citet[p. 154]{ledyard1997public}, model players and voters as moving towards the \emph{optimal} social outcome, while we model voters that move towards \emph{a given} social outcome.
Their models have slightly more flexiblility on the weight of social outcomes as they analyze worst-case (distortion, price of anarchy) results in their respective settings, while we focus on \emph{average-case} outcomes, which requires the additional assumption that all voters weigh the external information equally.\footnote{Average-case voting analyses are based on the premise that votes are i.i.d. samples from some distribution; changing $\alpha$ induces some unique distribution shift per voter, and therefore samples are no longer identically distributed.}
\citet{bedaywi2023distortion} use the same model of voting as \citet{flanigan2023distortion} to bound the distortion of participatory budgeting mechanisms with public-spirited voters.

\paragraph{Average-case voting}
While the majority of the computational social choice literature focuses on worst-case impossibilities and tradeoffs, a solid line of work has examined average-case voting outcomes.
These works tend to examine voting outcomes in expectation over i.i.d. votes from some density $\mu$.
When $\mu$ is the uniform density, this yields the \emph{impartial culture} assumption, which is notably unrealistic, but an important starting point for understanding average-case voting~\citep{lehtinen2007unrealistic,tsetlin2003impartial}.
\citet{boutilier2012optimal} and \citet{apesteguia2011justice} both study the expected social welfare in average-case voting under slightly more general assumptions than impartial culture.
Both show that, if the impartial culture assumption is satisfied, then Borda voting maximizes the expected social welfare conditioned on the fixed ordinal votes.
Our approach to studying average-case voting is quite different; instead of examining expected social welfare conditioned on a voting profile, we are interested in changes to expected social welfare under mild, but directed, perturbations to voters' utilities.

\paragraph{Metric voting}
Finally, our voting model is similar to, but subtly different from metric voting, in which voters and alternatives are embedded into $\reals^d$, then preferences are computed by evaluating relative distances to alternatives.
Distortion bounds in metric voting models are well understood~\citep{anshelevich2018approximating,skowron2017social,charikar2023breaking}.
In our model, instead of voters and \emph{alternatives} being embedded into $\reals^d$, as in metric voting, we assume voters and \emph{reports} are embedded into $X \subseteq \reals^m$, and voters submit the report they are closest to in distance.

\section{Voting Model}
Suppose $n$ voters are working together to select an alternative in $[m] := \{1,2,\ldots, m\}$.
Voter $i$ has some normalized utility $u_i \in \simplex$, where $\simplex = \{u \in \reals^m_+ : \|u\|_1 = 1\}$ denotes the probability simplex.
We suppose utilities are drawn from some density $\mu$ defining the probability space $([m], \simplex, \mu)$.
Voters must select some score vector $r \in \R$ to assign to alternatives such that $|\R| = R < \infty$, where $\R \subset c \simplex$ for some $c > 0$ is a subset of a general simplex.
(Up to scaling, we can intuitively think of $\R$ as being a subset of $\simplex$.)
At times, we abuse notation and write $\R$ as a matrix in $\reals^{m \times R}$.
A set of votes is then input to the voting rule $f : \N^R \to [m]$ which maps frequencies of assigned score vectors to one selected alternative.
If $f$ is anonymous, then for each $r \in \R$, every permutation $\sigma(r) \in \R$ as well.
We often work with two $\R$: for Plurality, we let $\R = \Rplurality := \{e_i : i \in [m]\}$ be the standard basis, and when eliciting fully ordinal rankings, we let $\R = \Rborda := \{\sigma([m])\}$ be all permutations of $[m]$.

We suppose voter $i$ has preferences $u_i \sim \mu$, and voter $i$ chooses to report the score $r \in \R$ minimizing $d(u_i, r)$, where $d(\cdot, \cdot)$ is Euclidean distance.
The ``scoring menu'' $\R$ induces the function $\GammaR: u \mapsto \argmin_{r \in \R} d(u,r)$ mapping real-valued utilities to reports.
If $\R$ is understood from context, we simply denote this function $\Gamma$.
We suppose ties are broken uniformly at random between relevant parties.
Often, we work with the level set $\Gamma_r = \{u \in \simplex : r \in \Gamma(u)\}$, which is the set of cardinal preferences leading to a voter submitting $r \in \R$.
Each kite in Figure~\ref{fig:plurality-ind} represents a different level set for $\Rplurality$.

\begin{example}[Plurality voting]
Consider the Plurality voting rule over $m = 3$ alternatives: $a$, $b$, and $c$.
A vote for alternative $a$ ascribes scores $(1,0,0)$ to alternatives $(a,b,c)$ respectively; a vote for alternative $b$ ascribes $(0,1,0)$, and a vote for $c$ score $(0,0,1)$.
The level set $\Gamma_{(1,0,0)} = \{u \in \Delta_3 \mid u_1 \geq u_2 \text{ and } u_1 \geq u_3\}$ is then the set of utility profiles representing voters who most prefer alternative $a$.
For a voter with cardinal utility $(\frac 1 5, \frac 1 2, \frac 3 {10})$, we have their vote $\Gamma((\frac 1 5, \frac 1 2, \frac 3 {10})) = \{(0,1,0)\}$, meaning they vote for alternative $b$.
Equivalently, their utility $(\frac 1 5, \frac 1 2, \frac 3 {10}) \in \Gamma_{(0,1,0)}$ is in the $(0,1,0)$ level set.
See Figure~\ref{fig:plurality-ind} for depictions of the level sets for plurality and Figure~\ref{fig:borda-ind} for Borda and other voting rules that elicit full ordinal rankings over outcomes.
\end{example}

Since preferences are drawn from some unknown $\mu$, preferences are typically approximated by measuring $\frac{|\{i \mid \Gamma(u_i) = r\}|}{n} \approx \mu(\Gamma_r) = \int_{\Gamma_r} d\mu$.
Throughout, the measure of these level sets serve as a central piece of our analysis, as external information shifts the distribution $\mu$ in some directed way.
In general, we denote the probability distribution over reports with $p := \{\mu(\Gamma_r)\}_{r\in\R}$.

\begin{figure}[h]
\begin{minipage}{0.45\linewidth}
    \centering
\includegraphics[width=\linewidth]{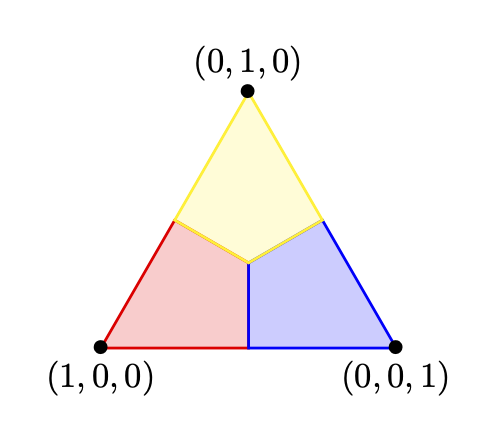}
    \caption{Plurality: Voters whose utilites lie in the red region (lower left) vote for alternative $a$ by ascribing score $(1,0,0)$, and similar for yellow voters voting for alternative $b$, and blue voting for alternative $c$.}
    \label{fig:plurality-ind}
\end{minipage}
\hfill
\begin{minipage}{0.45\linewidth}
    \centering
\includegraphics[width=\linewidth]{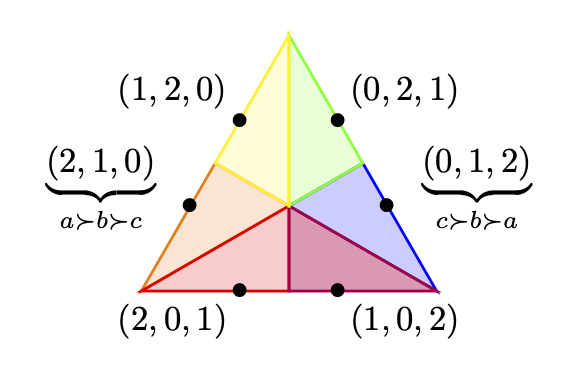}
    \caption{Fully ordinal prefs: Voters whose utility lies in the orange region vote $a \succ b \succ c$ by ascribing score $(2,1,0)$ and so on by permutations of $[m]$.}
    \label{fig:borda-ind}
\end{minipage}
\end{figure}

To reason about anchoring information, we suppose a mechanism designer knows this information $w \in \simplex$ and voters are influenced by weight $\alpha \in [0,1]$ such that, instead of voting according to their standard utilities $u$, voters instead vote for the report $r \in \R$ minimizing $d((1-\alpha) u_i + \alpha w,r)$.
That is, each voter shifts their preferences $\alpha$ towards $w$.
We are interested in understanding the role of $w$ in changing voting outcomes, particularly its effect on social welfare.

\begin{example}[Changing vote via $w$]
    Suppose a voter has utility $(1/2, 9/20,1/20)$, meaning they very slightly prefer the first outcome the most, then the second, and have almost no utility for the third.
    If external information suggests $w = (0, 1/2, 1/2)$ and they are influenced with weight $\alpha = \frac 1 {10}$, then their updated preferences are $(0.45, 0.455, 0.095)$, and they switch their vote from candidate $a$ to candidate $b$.
\end{example}

Given a profile of $n$ utilities $\bld u \in \simplex^n$, we consider the induced histogram of votes $\#\Gamma(\bld u) \in \N^R$, which are determined by cardinal utility profiles $\bld u$ and the (property induced by the) report set $\GammaR$. 
The $r^{th}$ entry of the histogram represents the number of voters submitting score $r$.\footnote{We generally assume preferences are drawn from a smooth enough distribution $\mu$ that the probability of a voter being exactly indifferent between two reports is $0$.}
We denote a general histogram of votes by $h \in \mathbb{N}^R$.

\begin{example}[Induced histograms]
    Consider the Borda voting rule with the report set $\R$ being all $m!$ permutations of $\{0, 1, \ldots, m-1\}$. 
    If $m = 3$ and $n = 15$, a voting profile $h = (1,4,2,2,5,1) \in \mathbb{N}^{m!}$ implies one voter prefers $a \succ b \succ c$, four voters prefer $a \succ c \succ b$, etc., fixing the order of $\R$.
    If the voting rule is Plurality, then $h \in \mathbb{N}^m$, and $h = (5,4,6)$ means $5$ voters have $a$ as their most preferred candidate, $4$ have $b$, and $6$ have $c$ as their most preferred.
\end{example}

In \S~\ref{sec:change-sw}, we examine positional scoring rules $f : \mathbb{N}^R \to [m]$, which are those that select the alternative with highest score given a histogram of score assignments $h \in \mathbb{N}^R$, denoted $f(h) = \argmax_{a} (\inprod{\R}{h})_a$.
If $|f(h)| >1$, there is a tie, which we assume are broken uniformly at random, as is common in the literature (cf.~\citep{boutilier2012optimal}).
Given a report set $\R$ and $n$ voters, let $f_a = \{h \in \mathbb{N}^R \mid a \in f(h), \|h\|_1 = n\}$, which is the set of histograms which lead to alternative $a$ being selected by voting rule $f$.

\begin{example}[Determining the winner of an election]
    Suppose $m = 3$ and fix $n$.
    For alternative $a$, the set of histograms selecting alternative $a$ according to Plurality voting, $\fplurality_a = \{h \in \mathbb{N}^m \mid h_{(1,0,0)} \geq h_{(0,1,0)} \text{ and } h_{(1,0,0)} \geq h_{(0,0,1)}\}$ is the set of histograms where $a$ receives more votes than $b$ and $c$.
    If the voting rule is veto, then the histograms selecting alternative $a$ $\fveto_a = \{h \in \mathbb{N}^R \mid h_{(1,0,1)} \geq h_{(0,1,1)} \text{ and } h_{(1,1,0)} \geq h_{(0,1,1)} \}$ are those that veto $a$ less than $b$ and $c$.
\end{example}

\section{Technical tools to analyze how voting outcomes change}
In order to understand the effect of anchoring information on voting outcomes, we need two technical tools: first, we need to derive bounds on the probability any fixed alternative is selected when voters submit reports according to \emph{standard preferences}, so that we can later juxtapose these bounds with anchored preferences (\S~\ref{sec:sw-standard}).
This tool is used in \S~\ref{sec:sw-social-tighten} in conjunction with the mechanism in \S~\ref{sec:elicit-social} in order to show that both upper and lower bounds on probability of the alternative $\astar$ most preferred by $w$ is selected increase.
The second technical tool we need is an intermediary mechanism that elicits anchored preferences (\S~\ref{sec:elicit-social}).
This mechanism modifies the report scores $\R$ into a \emph{hypothetical set of scores} $\M$ so that voting under $u$ with options in $\M$ always yields the same vote as voting under $(1-\alpha) u + \alpha w$ with options in $\R$.
This enables us to examine and compare the utility profiles $u$ with the submitted report changes with the introduction of anchoring information.

\subsection{Bounding the probability alternative $a$ is selected via balls-and-bins}\label{sec:sw-standard}
We first examine the probability of different outcomes being reached by a few standard voting rules.
We do this by bounding the probability any given alternative $a$ is selected as the outcome. 
If, for a fixed voting rule $f$, there is some function $c : \mathbb{N} \to \mathbb{N}$ such that, for all alternatives $a \in [m]$, there exists a set of reports $Q_a \subseteq \R$ such that at least $c(n)$ reports belonging to the set $Q_a$ implies that $a$ selected by the voting mechanism.
That is, $\sum_{r \in Q_a} h_r > c(n) \implies f(h) = \{a\}$ is sufficient for a vote to be determined solely by the finite set $Q_a$ of reports, rather than searching over the entire space of votes.
(For example, if $c(n) = \frac n 2$, this reduces to the majority criterion~\citep[p. 266]{pennock1977due}.)
In particular, with the report set $\Rborda$, we often use the set $Q^\sigma_a := \{r \in \Rborda : a \text{ is highest ranked alternative}\}$.

With this sufficient condition, we bound the probability alternative $a$ is selected by a reduction to the nonasymptotic balls and bins setting\footnote{Throwing $n$ balls into $R$ bins that land in each bin with probability according to $p$, what is the probability the max-loaded bin has $\geq k$ balls inside? e.g., \citep{mitzenmacher2001power}} studied by~\citet{schultegeers2022balls}. In \S~\ref{sec:sw-social-tighten}, we show that if voters have anchored preferences, these lower bounds tighten for $w$'s most preferred alternative $\{\astar\} = \argmax_{a'}w_{a'}$, which demonstrates the influence of the anchoring point $w$.
Throughout this section, we let $p := \left\{\mu(\Gamma_r)\right\}_r$ be the probability distribution over coarse reports $\R$.

\paragraph{Plurality}
First, we use the fact that plurality satisfies the majority criterion as a sufficient condition for an alternative $a$ to be selected by $\fplurality$, which we reduce to balls and bins.
Moreover, by the Mean Value Theorem, we know that some alternative must receive $n / m$ votes in Plurality voting, with which we obtain upper bounds on the probability of any alternative $a$ winning an election.

\begin{propositionE}[][end,restate]\label{prop:plurality-standard-lb}
Given a density $\mu$, consider $p~:=~\left\{\mu(\Gamma_r)\right\}_r$.
    For any $a \in [m]$, we have 
    \begin{multline*}
        \min\left({n \choose \frac{n}{m}}\|p_a\|^{n/m}_{n/m}, 1 - \sum_{b \neq a} \left(\Pr[Binom(n, p_b) > n/2]\right) \right)\\ \geq \Pr_{\bld u}[\fplurality(h) = a] 
        \geq \Pr[Binom(n, p_a) \geq n/2].
    \end{multline*}
\end{propositionE}
\begin{proofE}
For the lower bound,
    \begin{align*}
        \Pr_{\bld u}[\fplurality(h) = a] &\geq \Pr_{\bld u}[h_a > n/2] + \frac 1 2 \Pr_{\bld u}[h_a = n/2] & \text{Majority criterion} \\
        &\geq \Pr_p[Binom(n, p_a) > n/2] + \frac 1 2 \Pr_p[Binom(n, p_a) = n/2] & \text{\citep{schultegeers2022balls}}~.
    \end{align*}

    For the upper bound, observe that the second term is simply $1-$ the probability any other alternative is selected.

    The first term leverages the mean value theorem.
    In plurality voting, precisely $n$ points are distributed among the $m$ alternatives, so some alternative must receive at least $\frac n m$ votes.
    This also reduces to the nonasymptotic balls-and-bins, and we leverage the upper bound from~\citet[Theorem 1]{schultegeers2022balls}.
\end{proofE}
Proposition~\ref{prop:plurality-standard-lb} gives both upper and lower bounds on the probability that the Plurality voting rule selects alternative $a$.
Lower bounds are determined by the probability a Binomial sample of $n$ votes returns at least $c(n) := n/2$ ``yes'' results, and the upper bound is given by either the probability some other alternative receives at least $n/2$ votes or the probability that alternative $a$ receives at least $\frac n m$ votes; a necessary condition for winning the election.

\paragraph{Borda}
We derive a similar bound for Borda, though this bound is strictly looser than the bound for Plurality, as the number of votes required for the sufficient condition is $c(n) = \frac{n(m-1)}{m}$: strictly greater than the $\frac n 2$ required for Plurality.
For large $m$, this bound loses its meaning as $c(n) \overset{m \to \infty}{\to} n$, but it is helpful for small $m$.

\begin{proposition}\label{prop:borda-standard-lb}
    For any $a \in [m]$, density $\mu$, and $c(n) := \frac{n(m-1)}{m}$, then 
    \begin{align*}
        1 - \sum_{b \neq a} \left(\Pr[Binom(n, p_{Q^\sigma_b}) > c(n)] \right) &\geq \Pr_{\bld u}[\fborda(h) = a] \\
        &\geq \Pr_{\bld u}[Binom(n, \|\sum_{r \in Q^\sigma_a}p_r\|_{c(n)}) > c(n)]~.
    \end{align*}
\end{proposition}

The proof follows in the same manner as the proof of Proposition~\ref{prop:plurality-standard-lb}.
Intuitively, because Borda rewards alternatives for being highly, but not highest, ranked.
This makes it easier for polarizing alternatives ``to lose'' an election than plurality, and Borda makes it easier ``to win'' by generally appealing to most voters for less polarizing alternatives.
This means the bounds are weaker than those given by plurality since these sufficient conditions are on ``winning'' and election, rather than ``not losing.''

\paragraph{Other voting rules}
Like Plurality, Copeland and Instant Runoff Voting (IRV) also satisfy the majority criterion.
Therefore, if, for all $a$, there is a set $Q_a \subseteq \Rborda$ such that $\sum_{r \in Q_a} h_r > \frac n 2 := c(n)$, then $f^{Copeland}(h) = f^{IRV}(h) = \{a\}$, then similar bounds can be derived.

Above, we give bounds based on sufficient conditions for an alternative to be declared the winner.
However, in Veto voting, we can intuitively think of the winner as being the ``least disliked'' alternative, which does not lend itself to (nontrivial) sufficient conditions that are independent of other alternatives.

\subsection{Eliciting anchored preferences}\label{sec:elicit-social}

In practice, algorithm designers do not know the density of voters' utilities $\mu$, and instead approximate the probability a voter submits report $r$ by considering the measure of the report set $\mu(\Gamma_r)~=~\int_{\Gamma_r}~d\mu~\approx \frac{|\{i : \Gamma(u_i) = r\}| }{n}$ with only access to $n$ reports $\{\Gamma(u_i)\}_{i \in [n]}$ instead of cardinal utilities $\{u_i\}_{i \in [n]}$.
This yields a more practical representation of the density $\mu$, but sacrifices the granular information needed to understand the impacts of the anchoring point $w$.

In Theorem~\ref{thm:elicit-social} show there is another ``scoring menu'' $\M$ inducing a new property $\Gaw$ and bijection $\phi : \R \to \M$ such that $\M$ simulates the role of anchoring information, demonstrated with the equivalence $(1-\alpha) u + \alpha w \in \Gamma_r \iff u \in \Gaw_{\phi(r)}$ for all $r \in \R$.
That is, for every report $r \in \R$, an anchored vote with the standard mechanism yields the report $r$ if and only if a non-anchored vote under the perturbed menu also produces vote $r$.
This (non-anonymous) scoring rule emulates each voter moving $\alpha$ towards the anchoring point $w$ without eliciting their full cardinal utility $u_i$.
We elicit slightly more granular information than standard ordinal preferences, but do not require full cardinal preferences.

The proof of Theorem~\ref{thm:elicit-social} is inspired from previous results on property elicitation~\citep{lambert2009eliciting,lambert2018elicitation}.

\begin{theoremE}[][normal]\label{thm:elicit-social}
    Given a scoring menu $\R$ inducing $\GammaR$, a voter with anchored preferences (parameterized by $w, \alpha$) will vote the same way under the property $\Gaw$ induced by the scoring menu $\M := \{\frac{r - \alpha w}{1-\alpha}\}_{r}$.
    That is, $(1-\alpha)u + \alpha w \in \Gamma_r \iff u \in \Gaw_{\phi(r)}$ for all $r \in \R$, where $\phi : r \mapsto \frac{r-\alpha w}{1-\alpha}$.
\end{theoremE}
\begin{proofE}
    Let $v = (1-\alpha)u + \alpha w$.
    For a score $r$, we have
    \begin{align*}
        v \in \Gamma_r &\iff d(v, r) \leq d(v, {r'}) \qquad \forall r' \neq r \\
        &\iff \|r - v\|^2 \leq \|r' - v\|^2 \qquad \forall r' \neq r \\
        &\iff \|\frac{r - \alpha w}{1- \alpha} - u\|^2 \leq \|\frac{r' - \alpha w}{1-\alpha} - u\|^2 \qquad \forall r' \neq r \\
        &\iff \|u - \frac{{r} - \alpha w}{1-\alpha}\|^2\leq \|u - \frac{{r'} - \alpha w}{1-\alpha}\|^2 \qquad \forall r' \neq r\\
        &\iff u \in \Gaw_{\phi(r)}~.
    \end{align*}
\end{proofE}
Theorem~\ref{thm:elicit-social} shows that we can use $\M$ as a set of scoring vectors which simulates the addition of external information.
This allows us to conceptually replace the shift in votes with a shift in the voting rule, and avoid trying to compute a distribution shift for a distribution we do not even know.
We benefit as we can now equivalently use the histogram of standard utilities when voting according to $\M$, the histograms $\#\Gaw(\bld u)= \#\Gamma^{\phi(\R)}(\bld u)$, which is equal to the histogram of anchored preferences according to standard scores $\#\GammaR((1-\alpha) \bld u + \alpha w)$ by Theorem~\ref{thm:elicit-social}.
Deriving $\M$ then allows us to empirically estimate the distribution $q := \{\mu(\Gaw_s)\}_{s \in \M}$ when $\mu$ is not known exactly, and does not require eliciting voters' cardinal preferences.
In the sequel, we slightly abuse notation and let $\Gaw_r$ denote what is technically $\Gaw_{\phi(r)}$, which is the set of utilities which would lead voters to submit report $r$ under \emph{anchored preferences}.
Since $\phi$ is defined by $\alpha$, the mechanism designer is the entity imparting the value of anchored preferences $\alpha$ rather than the voters.\footnote{If $\alpha$ varies for each voter and menus were customized, we can no longer use the assumption that voters are identically distributed. Intuitively, we think of $\alpha$ as an average of individual $\alpha_i$s if $\alpha_i$ is drawn from a single-peaked and roughly symmetric distribution.}

\subsubsection{Anchored votes align more closely with $w$ on the individual level}\label{sec:individual-social-w}
We now ask, if one elicits anchored preferences, how might an individual's vote change, given a general anchoring point $w \in \simplex$.
Perhaps unsurprisingly, we show that individual votes only move to reports that ``better align'' with $w$ in this model.

\begin{propositionE}[][end,restate]\label{prop:move-up-sw}
Fix $w \in \simplex$ and $\alpha \in [0,1)$, and consider the function $\phi : r \mapsto \frac{r - \alpha w}{1-\alpha}$.
    For any two reports $s,t \in \R$, for voter utility $u \in \simplex$ such that $d(u,s) \leq d(u,t)$ and $\langle w, s \rangle \geq ]langle w, t \rangle$, then $d(u, \phi(s)) \leq d(u, \phi(t))$.
\end{propositionE}
\begin{proofE}
    It suffices to show the inequality on squared distance.
Observe that we can re-write the squared distance $d(u, \phi(s))^2 =  d(u, \frac s {1-\alpha})^2 + \sum_{i \in [m]} \frac{\alpha w_i}{1-\alpha} (\frac{\alpha w_i + 2 u_i -2 s_i}{1-\alpha})$.
    \begin{align*}
        \sum_{i \in [m]} w_i (t_i - s_i) &\leq 0  \\
        \sum_{i \in [m]} \frac{\alpha w_i}{1-\alpha} \left(\frac{\alpha w_i + 2u_i -2 s_i}{1-\alpha} - \frac{\alpha w_i +2u_i-2 t_i}{1-\alpha}\right) &\leq 0 \\
        \sum_{i \in [m]} \frac{\alpha w_i}{1-\alpha} \left(\frac{\alpha w_i + 2u_i -2 s_i}{1-\alpha}\right) &\leq \sum_{i \in [m]} \frac{\alpha w_i}{1-\alpha} \left(\frac{\alpha w_i +2u_i-2 t_i}{1-\alpha}\right) \\ 
        \implies d\left(u, \frac{s}{1-\alpha}\right)^2 + \sum_{i \in [m]} \frac{\alpha w_i}{1-\alpha} \left(\frac{\alpha w_i + 2u_i -2 s_i}{1-\alpha}\right) &\leq d\left(u, \frac{t}{1-\alpha}\right)^2 + \sum_{i \in [m]} \frac{\alpha w_i}{1-\alpha} \left(\frac{\alpha w_i + 2u_i -2 t_i}{1-\alpha}\right) \\
        d(u, \phi(s))^2 &\leq d(u, \phi(t))^2
    \end{align*}
\end{proofE}
Proposition~\ref{prop:move-up-sw} suggests that reports ``aligning with $w$'' increase in probability: if an agent's vote moves from $s$ to $t$, then $w$ ``prefers $s$ to $t$'', given by $\langle w, s \rangle < \langle w, t \rangle$.
Wquivalently, no player will move from $s$ to $t$ if  $w$ prefers $s$ to $t$, again given $\inprod w s > \inprod w t$.
This phenomena is demonstrated in Figure~\ref{fig:plurality-social}, where the red cell, corresponding to the vote for alternative $a$, which is most preferred by $w$, is bigger under anchored voting. 
Proposition~\ref{prop:move-up-sw} immediately implies Corollary~\ref{cor:preserve-order-w}, which states that the set of cardinal utilities yielding the report $r^\star$ that ``would be submitted by $w$'' strictly increases.

\begin{corollary}\label{cor:preserve-order-w}
    Let an anchoring point $w \in \simplex$ and weight $\alpha \in (0,1]$ be given and consider the report $\{r^\star\} = \argmax_{r' \in \R} \langle w, r' \rangle$. 
    Then $\GammaR_{r^\star} \subsetneq \Gaw_{r^\star}$.
\end{corollary}

\begin{figure}[h]
\begin{minipage}{0.45\linewidth}
    \centering
	\includegraphics[width=\linewidth]{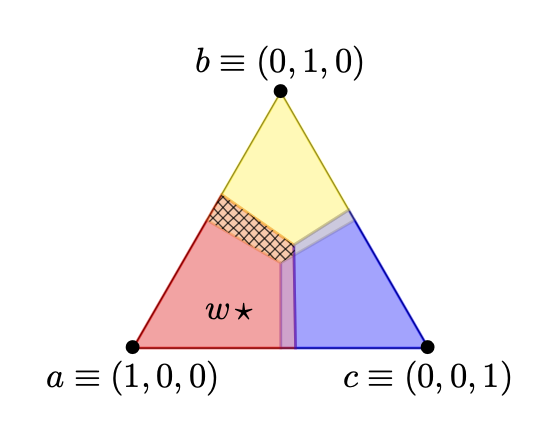}    
    \end{minipage}
    \hfill
    \begin{minipage}{0.48\linewidth}
    \caption{Anchoring Plurality votes in the point $w$ yields a menu $\M$ with the level sets of $\Gaw$ laid over the level sets of $\GammaR$. Utilities in red and yellow intersection (hatched) will vote for candidate $b$ under standard preferences and candidate $a$ with the anchoring information. They are close to indifferent between $a$ and $b$ under standard preferences.}
\label{fig:plurality-social}
    \end{minipage}
\end{figure}

Under the common \emph{impartial culture} assumption that $\mu$ is the uniform density $m$, this further implies that the probability of any anchored preference is ordered according to $w$.
(Recall that we do not generally requite the impartial culture assumption, unlike much of the previous literature.)

\begin{corollary}\label{cor:preserve-order-w-uniform}
    Let $f$ be an anonymous voting rule.
    If $m$ is the uniform density, then for all $s,t \in \R$, the probability an agent reports $s$ is greater than the probability they report $t$ if and only if $w$ prefers $s$ to $t$. 
    That is, $m(\Gaw_{s}) \geq m(\Gaw_t) \iff \langle w, s \rangle \geq \langle w, t \rangle$. 
\end{corollary}

This implies that, when $\mu$ is uniform, incorporating external information anchored in $w$ leads to the probabilities of alternatives being elected to follow the same order as $w$ if we recursively apply this argument to all pairs of reports $s,t$.
While mathematically tractable, it is well-known that impartial culture assumptions are often unrealistic~\citep{lehtinen2007unrealistic}.

\section{Anchored preferences can increase the probability $w$'s preferred alternative is selected}\label{sec:sw-social-tighten}
Recall that in \S~\ref{sec:sw-standard} we reduce voting in expectation to a balls-to-bins problem to obtain lower and upper bounds on the probability an alternative $a$ is selected.
We naturally wonder if and when anchored preferences tighten these lower bounds or loosen the upper bounds.
We show in Corollary~\ref{cor:social-prefs-tighten-lb} that, for the alternative $\astar$ such that $\{\astar\} = \argmax_{a'}w_{a'}$, the lower bounds do indeed tighten when the voting rule is Plurality, and give conditions on $w$ that tighten the lower bounds for Borda.
Moreover, under certain smoothness conditions on $\mu$, we also show the upper bounds loosen (increase) in Corollary~\ref{cor:social-prefs-loosen-ub}.

Since the bounds in \S~\ref{sec:sw-standard} are given by the binomial distribution, we first study the monotonicity of the binomial distribution in the probability of a positive event.
We show in Lemma~\ref{lem:binom-inc} that if the probability of a positive instance increases, then the probability that a binomial sample has at least $k$ positive instances (meaning $r \in Q_a$) also increases.
For plurality, this suffices to tighten bounds, and in Lemma ~\ref{lem:suff-cond-borda-w}, we give a sufficient condition on $w$ for the probability that $r \in Q_{\astar}$ to increase with anchored preferences.

\begin{lemmaE}[Monotonicity of binomial bounds in $p$][end,restate,text link]\label{lem:binom-inc}
    Suppose $p \leq q$.
    Then for any $k \in [n]$, we have $\Pr[Binom(n, p) \geq k] \leq \Pr[Binom(n,q) \geq k]$.
\end{lemmaE}
\begin{proofE}
    It suffices to show $\frac d {dp}\Pr[Binom(n, p) \geq k] \geq 0$ for all $p \in [0,1]$.
    We recursively derive the closed form of $\frac d {dp}\Pr[Binom(n, p) \geq n - k] = (n-k) {n \choose k} p^{n-k-1} (1-p)^k$, which is nonnegative for all $p \in [0,1]$, yielding the result.
    In general, we have the form
    \begin{align*}
        \frac d {dp}\Pr[Binom(n, p) \geq n - k] 
        &= \sum_{\ell=0}^k\frac d {dp} \Pr[Binom(n,p) = n-\ell] \\
        &= \sum_{\ell = 0}^k {n \choose \ell} p^{n-\ell - 1} (1-p)^{\ell-1} (n-\ell - np)~.\\
    \end{align*}
    First, let us consider the base case: $k = 0$.
    \begin{align*}
        \frac d {dp}\Pr[Binom(n, p) \geq n] &=  p^{n- 1} (1-p)^{-1} n (1-p)\\
        &= n p^{n-1}\\
        &= (n-0) {n \choose 0} p^{n-1} (1-p)^0
    \end{align*}
    Now assume this holds for $k$. We show it holds for $k+1$.
    \begin{align*}
        & \frac d {dp}\Pr[Binom(n, p) \geq n-(k+1)] \\
        &= \frac d {dp}\Pr[Binom(n, p) \geq n-k] + \frac d {dp}\Pr[Binom(n, p) = n-(k+1)] \\
        &= \overbrace{{n \choose k}(n-k)p^{n-k-1} (1-p)^k} + \overbrace{{n \choose k+1}p^{n - (k+1)-1} (1-p)^k (n - (k+1) - np)} \\
        &= {n \choose k+1}p^{n-k-2} (1-p)^k (p (k+1)) + {n \choose k+1}p^{n - (k+1)-1} (1-p)^k (n - (k+1) - np) \\
        &= {n \choose k+1}p^{n- k- 2} (1-p)^{k} (p(k+1) + (n - k - 1 - np)) \\
        &= {n \choose k+1}p^{n- k- 2} (1-p)^{k} ((1-p) (n-k - 1)) \\
        &= {n \choose k+1}(n-k-1)p^{n- k- 2} (1-p)^{k+1} 
    \end{align*}
    
\end{proofE}
One corollary of Proposition~\ref{prop:move-up-sw} is that when the report set is the set of elementary basis vectors, $\R = \Rplurality$, then $\mu(\GammaR_{\astar}) \leq \mu(\Gaw_{\astar})$ for all $w$ and $\alpha$. In turn, Lemma~\ref{lem:binom-inc} tightens the lower bounds from Proposition~\ref{prop:borda-standard-lb} for all $w \in \simplex$. 

However, when $\R = \Rborda$, this is not necessarily the case.
For intuition, if $w$ is close to indifferent between two alternatives as its preferred, then there might be some report $s \not\in Q_{\astar}$ that does not rank $\astar$ as its top choice, yet because of utility over non-top outcomes, we might have $\inprod w s \geq \inprod w t$ for some $t \in Q_{\astar}$.
This opens up the possibility that the probability a voter ascribes the highest score to alternative $\astar$ decreases, i.e., $\sum_{r \in Q_{\astar}} \mu(\Gamma_r) > \sum_{r \in Q_{\astar}} \mu(\Gaw_r)$, as votes might move from $t$ to $s$.
In Lemma~\ref{lem:suff-cond-borda-w}, we give a sufficient condition on $w$ for the binomial lower bound to be tightened, illustrated with $m=3$ alternatives in Figure~\ref{fig:borda-w-tighten-condition}.
Intuitively, if $w$ favors one alternative enough, then the reports $r \in \R$ that most align with $w$ are those that assign the highest score to $a$.
In the sequel, let $k = (m-1)!$.

\begin{lemmaE}[Sufficient condition on $w$ for probability of a report in $Q^\sigma_{\astar}$ to increase][end,restate]\label{lem:suff-cond-borda-w}
    Let $\R = \Rborda$.
    If $Q^\sigma_{\astar} = \topk{k}(\langle w, r\rangle_{r \in \R})$, then $\sum_{r \in Q^\sigma_{\astar}}\mu(\Gaw_{\phi(r)}) \geq \sum_{r \in Q^\sigma_{\astar}}\mu(\GammaR_r)$, and in turn, $q_{Q^\sigma_{\astar}} \geq p_{Q^\sigma_{\astar}}$.
\end{lemmaE}
\begin{proofE}
    By Proposition~\ref{prop:move-up-sw}, for any $r \in Q_{\astar}$, if $u \not \in \Gaw_{\phi(r)}$, it must be in $\Gaw_{\phi(t)}$ for some other $t \in Q_{\astar}$.
    Observe
    \begin{align*}
        \cup_{r \in Q_{\astar}} \GammaR_r &= \cup_{r \in Q_{\astar}, s \in \R} (\GammaR_r \cap \Gaw_{\phi(s)}) & \text{disjoint sets} \\
        &= \cup_{r \in Q_{\astar}, s \in Q_{\astar}} (\GammaR_r \cap \Gaw_{\phi(s)}) &\text{if $s \not \in Q_{\astar}$, then $\GammaR_r \cap \Gaw_{\phi(s)} = \emptyset$} \\
        &\subseteq \cup_{s \in Q_{\astar}, r \in \R} (\GammaR_r \cap \Gaw_{\phi(s)}) & Q_{\astar} \subseteq \R \\
        &= \cup_{r \in Q_{\astar}} \Gaw_r & 
    \end{align*}
    Since $\cup_{r \in Q_{\astar}} \GammaR_r \subseteq \cup_{r \in Q_{\astar}} \Gaw_{\phi(r)}$, we have $\sum_{r \in Q_{\astar}} \mu(\GammaR_r) = \mu(\cup_{r \in Q_{\astar}} \GammaR_r) \leq \mu(\cup_{r \in Q_{\astar}} \Gaw_{\phi(r)}) = \sum_{r \in Q_{\astar}} \mu(\Gaw_{\phi(r)})$.
\end{proofE}
Lemma~\ref{lem:suff-cond-borda-w} shows that if reports assigning score $m-1$ to alternative $\astar$ are the top-$k$ preferred by $w$, then the probability that any one vote assigns score $m-1$ to alternative $a$ increases.

\begin{lemmaE}[][all end,restate]\label{lem:top-k-borda}
    Fix $w \in \simplex$.
    Let $\astar = \argmax_{a'} w_{a'}$.
    If $w_a \geq (m-1) w_{[2]} + \sum_{i=3}^m w_{[i]} (m - 2(i-i))$, then $Q_a = \topk{k}(\{\langle w, r\rangle\}_{r \in \R})$.
\end{lemmaE}
\begin{proofE}
For ease of exposition, let $x_{[i]}$ denote the $i^{th}$ largest element of $x$.
    It suffices to show 
    \begin{align*}
        \min_{r \in Q_a} \inprod w r &\geq \max_{s \not \in Q_a} \inprod w s \\
        w_a (m-1) + \sum_{i=2}^m w_{[i]} (i-2) &\geq w_a (m-2) + w_{[2]} (m-1) + \sum_{i=3}^m w_{[i]} (m-i) \\
        w_a &\geq w_{[2]}(m-1) + \sum_{i=3}^m w_{[i]} (m-2i + 2)~,
    \end{align*}
    which concludes the proof.
\end{proofE}
With $m=3$, Lemma~\ref{lem:suff-cond-borda-w} shows that it suffices for $w_b \leq 1/3$ for all $b \neq a$ in order for the probability that alternative $\astar$ is the top-ranked to increase.
Lemma~\ref{lem:top-k-borda} in \S~\ref{app:omitted-proofs} characterizes $w$ such that $q_{Q^{\sigma}_{\astar}} \geq p_{Q^{\sigma}_{\astar}}$, the condition of which is captured by Assumption~\ref{assum:w-inc-borda}.
Throughout, we let $w_{[i]}$ denote the $i^{th}$ largest component of the vector $w$.

\begin{figure}
    \centering
    \begin{minipage}{0.4\linewidth}
    \includegraphics[width=\linewidth]{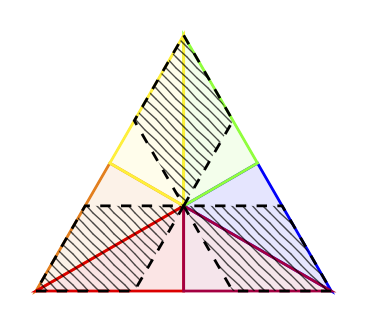}
    \end{minipage}
    \hfill
    \begin{minipage}{0.5\linewidth}
    \caption{Hatched regions subsets of $\Delta_3$ where any $w$ in the hatched region increases ensures that reports where $\astar$ is the favorite are precisely the top-$(m-1)!$ most preferred.  
    In turn, for any $w$ in the hatched regions, the bounds on the probability that $w$'s most favored alternative is selected increase.
}
    \label{fig:borda-w-tighten-condition}
    \end{minipage}
\end{figure}

\begin{assumption}\label{assum:w-inc-borda}
    For $m \geq 3$, let $w$ be an anchoring point such that $w_{[1]} \geq w_{[2]}(m-1) + \sum_{i=3}^m w_{[i]} (m-2i + 2)$.
\end{assumption}

In essence, we assume the anchoring point $w$ favors one alternative over others strongly enough.
If this is true, the probability that a voter ranks that alternative the highest increases. 
This now enables us to observe that the lower bounds derived in \S~\ref{sec:sw-standard} tighten under anchored preferences. 

\begin{corollary}\label{cor:social-prefs-tighten-lb}
    Let $w$ be an anchoring point with $\{\astar\} = \argmax_{a'}w_{a'}$.
    The lower bound given by Proposition~\ref{prop:plurality-standard-lb} tightens for the alternative $\astar$.
    Moreover, if $\R = \Rborda$ and $w$ satisfies Assumption~\ref{assum:w-inc-borda}, then the lower bound from Proposition~\ref{prop:borda-standard-lb} tightens.
\end{corollary}

Corollary~\ref{cor:social-prefs-tighten-lb} shows that anchored preferences can, but do not necessarily, tighten the lower bounds on the probability of a voting rule selects the alternative $\astar$ most preferred by $w$.
Moreover, the upper bound loosens if $\mu$ is sufficiently smooth.

\begin{corollary}\label{cor:social-prefs-loosen-ub}
    Let $w$ be an anchoring point such that ${\astar} = \argmax_{a'}w_{a'}$.
    Suppose $\astar$ is the only alternative that increases in probability (i.e., $p_{Q_b} \geq q_{Q_b}$ for all $b \neq \astar$). 
    Then the upper bounds in Propositions~\ref{prop:plurality-standard-lb} and \ref{prop:borda-standard-lb} loosen.
\end{corollary}

\section{Changes in social welfare arising from anchored voting}\label{sec:change-sw}
While Proposition~\ref{prop:move-up-sw} implies that individuals will only vote more in line with $w$, it is well-known that Paretian and non-dictatorial voting rules cannot satisfy independence of irrelevant alternatives~\citep{arrow1951social}, and intuitively, this suggests that that voting according to anchored preferences may change votes to irrelevant alternatives.
While an individual vote might increase the score of an irrelevant alternative, it's unclear when and if that will actually change the alternative selected by the voting rule.
In order to understand this, we examine sets of voting profiles in histogram space.

For a fixed, finite set of $n$ voters, the set of possible induced histograms is finite, though large.
Given a density $\mu$ on individual utilities and the induced distribution $p := \left\{ \mu(\Gamma_r) \right\}$, weight $\alpha$, and anchoring point $w$, we examine $\Gaw$ and examine the histogram distributions $\nu$ and $\nusoc$ such that 
\begin{align*}
\nu(f_a) 
    &= \sum_{h : a \in f(h)} \frac{1}{|f(h)|}\Pr_{\bld u}[\#\GammaR(\bld u) = h] \\
&= \sum_{h : a \in f(h)} \frac{1}{|f(h)|} (\frac{n!}{\prod_r h_r!}) \prod_r p_r^{h_r}\end{align*}
and likewise for $\nusoc(f_a)$, substituting $\Pr_{\bld u}[\#\GammaR((1-\alpha) \bld u + \alpha w) = h]$, which is a function of $q := \{\mu(\Gaw_{\phi(r)})\}$.
The $\frac 1 {|f(h)|}$ term breaks ties uniformly at random

\subsection{Expected social welfare increases when external information reflects social preferences}

We are now equipped to ask how the anchoring point $w$ affects social welfare?
We give a sufficient condition for social welfare to increase in expectation in Theorem~\ref{thm:inc-SW}.
Then, in Proposition~\ref{prop:eps-uniform-sat-inc-SW} we show that if the density $\mu$ is smooth enough, then an anchoring point $w$ roughly according to the ``average utililty'' $\E_\mu u$, this sufficient condition is satisfied.
Let $\sw(f(h), \bld u) := \sum_i u_i(f(h))$ be the total utility attained by the voting outcome $f(h)$.
Throughout, we let $\Delta \sw(\bld u) := \sw(f(\#\GammaR((1-\alpha)\bld u + \alpha w), \bld u) - \sw(f(\#\GammaR(\bld u), \bld u))$, and we ask when $\E [\Delta \sw \mid w, \alpha] \geq 0$.
That is, when does social welfare increase (in expectation) under anchored preferences?

\begin{theoremE}[][]
\label{thm:inc-SW}
Let $w, \alpha, \mu$ be given.
    Suppose $v := \E_{u \sim \mu} u$, and let the set $\inc := \{a \in [m] : \nusoc(f_a) - \nu(f_a) \geq 0\}$ be the set of alternatives that weakly increase in probability of being selected under anchored preferences, and $\dec := [m] \setminus \inc$ be the set of alternatives that decrease.
    If $\max_{a \in \dec} v_a \leq \min_{a \in \inc} v_a$, then $\E_{\bld u}[\Delta \sw(\bld u) \mid w, \alpha] \geq 0$.
\end{theoremE}
\begin{proof}
    First, observe that since social welfare is additive and preferences are drawn iid, we can restrict our focus to $v$ instead of $\bld u$.
    \begin{align*}
        \E_{\bld u \sim \mu^n} [\Delta \sw(\bld u) \mid w,\alpha] &= \E_{\bld u \sim \mu^n}[\sum_{i \in [n]} \sum_{a \in [m]} u_i(a) \left(\Pr_{\bld u}[f(h(\bld u, \Gaw)) =  a \mid w] -  \Pr_{\bld u}[f((h(\bld u, \Gamma)) = a]\right) \\
        &= \E_{\bld u \sim \mu^n} [ \sum_{i \in [n]} \langle u_i, \nusoc(f) - \nu(f) \rangle] \\
        &= n \E_{u \sim \mu} [ \langle u, \nusoc(f) - \nu(f) \rangle] \\
        &= n \langle v, \nusoc(f) - \nu(f) \rangle  \geq 0 \iff 
\langle v, \nusoc(f) - \nu(f) \rangle \geq 0
    \end{align*}
    Now we argue that under the assumptions, that $\langle v, \nusoc(f) - \nu(f) \rangle \geq 0$.
    
    First, observe that since $\nu$ and $\nusoc$ are probability distributions (both subject to affine constraints), $\sum_a \nusoc(f_a) - \sum_a \nu(f_a) = 0$, and therefore 
    \begin{align}
    \sum_{a \in \inc} (\nusoc(f_a) - \nu(f_a)) &= - \sum_{a \in \dec}(\nusoc(f_a) - \nu(f_a)) &= \sum_{a \in \dec}(\nu(f_a) - \nusoc(f_a)).\label{eq:order-diff}
    \end{align}
This yields
    \begin{align*}
        \langle v, \nusoc(f) - \nu(f) \rangle \geq 0 &\iff  
        \sum_{a \in \inc} v_a (\nusoc(f_a) - \nu(f_a)) \geq \sum_{a \in \dec} v_a (\nu(f_a) - \nusoc(f_a))
        \end{align*}
        We can show the latter inequality, as
        \begin{align*}
        \sum_{a \in \inc} v_a (\nusoc(f_a) - \nu(f_a)) &\geq \min_{a \in \inc} v_a \sum_{a \in \inc} (\nusoc(f_a) - \nu(f_a)) \\
        &= \min_{a \in \inc} v_a \sum_{a \in \dec} (\nu(f_a) - \nusoc(f_a)) \qquad \text{Eq.~\eqref{eq:order-diff}}\\
        &\geq \max_{a \in \dec} v_a \sum_{a \in \dec} (\nu(f_a) - \nusoc(f_a)) \qquad \text{by assumption}\\
        &\geq \sum_{a \in \dec} v_a (\nu(f_a) - \nusoc(f_a))
    \end{align*}

\end{proof}
Intuitively, Theorem~\ref{thm:inc-SW} shows that expected social welfare increases if the anchoring point $w$ is at least a good approximation of the true expected utility $v$.

Moreover, if $\mu$ is approximately uniform (Definition~\ref{def:eps-uniform}), having small total variation distance (TV) from the uniform density, then the assumptions of Theorem~\ref{thm:inc-SW} are satisfied, yielding Proposition~\ref{prop:eps-uniform-sat-inc-SW}.

\begin{definition}[$\epsilon$-uniform density]\label{def:eps-uniform}
    A density $\mu$ is $\epsilon$-uniform if, $\sup_{A \in \mathcal{B}} \left| \int_A d\mu - \int_A dm \right| \leq \epsilon$.
    That is, $TV(\mu,m) \leq \epsilon$, where $m$ is the uniform density.
\end{definition}

Our notion of $\epsilon$-uniformness is a much stricter condition than necessary, but its generality allows us to draw conclusions about \emph{any} $w$ that is ordered roughly according to the ``average utility'' $\E_\mu u$.

\begin{lemmaE}[][all end]\label{lem:mu-same-order-wR}
    Let $f$ be an anonymous voting rule, and $\mu$ be $\epsilon$-uniform.
    If $2 \epsilon \leq \min_{s,t \in \R} |m(\Gaw_t) - m(\Gaw_s)|$, then $\{\mu(\Gaw_t)\}_t$ preserves the same order as $\langle w, \R\rangle$.
\end{lemmaE}
\begin{proofE}
    Since $\mu$ is $\epsilon$-uniform, then $|\int_{\Gaw_r} d\mu - \int_{\Gaw_r} dm|  = |\mu(\Gaw_r) - m(\Gaw_r)|\leq \epsilon$ implies $\mu(\Gaw_r) \in [m(\Gaw_r) - \epsilon, m(\Gaw_r) + \epsilon]$.
    As $\{m(\Gaw_r)\}$ has the same order as $\langle w, \R \rangle$ (by Corollary~\ref{cor:preserve-order-w}), the result follows immediately.
\end{proofE}

\begin{lemmaE}[][all end]\label{lem:nusoc-same-order-nusocdiff}
    Let $f$ be an anonymous voting rule and $\mu$ be $\epsilon$-uniform.
    If $2 \frac{\sum_{k=1}^n {n \choose k}(R\epsilon)^k}{m} \leq \min_{a,b} |\nu(f_a) - \nu(f_b)|$, then $\{\nusoc(f_a)\}_a$ preserves the same order as $\{\nusoc(f_a) - \nu(f_a)\}_a$.
\end{lemmaE}
\begin{proofE}
    Since $\mu$ is $\epsilon$-uniform, we have $|\mu(\Gamma_r) - m(\Gamma_r)| \leq \epsilon$, which implies $\mu(\Gamma_r) \in [\frac{1 - R\epsilon}{R}, \frac{1 + R \epsilon}{R}]$ for all $r \in \R$ since $f$ is anonymous.
    This yields
    \begin{align*}
        \sum_{h \in f_a} \frac{1}{|f(h)|} \frac{n!}{\prod_r h_r!} \prod_r(\frac{1 - R\epsilon}{R})^{h_r} &\leq \nu(f_a) \leq \sum_{h \in f_a} \frac{1}{|f(h)|} \frac{n!}{\prod_r h_r!} \prod_r(\frac{1 + R\epsilon}{R})^{h_r} \qquad \forall a \in [m] \\
        \frac{(1 - R\epsilon)^n}{m} &\leq \nu(f_a) \leq \frac{(1 + R\epsilon)^n}{m} \qquad \forall a \in [m] \\
        \frac 1 m + \frac{\sum_{k=1}^n {n \choose k}(-R \epsilon)^k}{m} &\leq \nu(f_a) \leq \frac 1 m + \frac{\sum_{k=1}^n {n \choose k}(R \epsilon)^k}{m} \qquad \forall a \in [m] 
    \end{align*}
    Therefore, $\nu(f_a) \in [\frac 1 m - \delta, \frac 1 m + \delta]$ for $\delta := \max\left(|\frac{\sum_{k=1}^n {n \choose k}(-R \epsilon)^k}{m}|, \frac{\sum_{k=1}^n {n \choose k}(R \epsilon)^k}{m}\right) = \frac{\sum_{k=1}^n {n \choose k}(R \epsilon)^k}{m}$.

    If $2 \delta$ is smaller than the smallest gap $|\nusoc(f_a) - \nusoc(f_b)|$, then the order of $\nusoc(f_a)$ is preserved.    
\end{proofE}

\begin{lemmaE}[][all end]\label{lem:wR-same-order-nusocdiff}
    Let $f$ be an anonymous positional scoring rule and $\mu$ be a $\epsilon$-uniform density for $\epsilon$ sufficiently small.
    Then if $\left\{\mu(\Gaw_r)\right\}_r$ has the same order as $\langle w, \R \rangle$, then $\left\{\nusoc(f_a) - \nu(f_a)\right\}$ has the same order as $w$. 
\end{lemmaE}
\begin{proofE}
    By Lemma~\ref{lem:nusoc-same-order-nusocdiff}, $\{\nusoc(f_a)\}$ has the same order as $\{\nusoc(f_a) - \nu(f_a)\}$.
    We show the contrapositive.
    Suppose $\{\nusoc(f_a) - \nu(f_a)\}$ has a different order then $w$, then so does $\{\nusoc(f_a)\}$.
    Then there exist $a,b \in  [m]$ such that $\nusoc(f_a) > \nusoc(f_b)$ but $w_a < w_b$.

    $\nusoc(f_a) > \nusoc(f_b) \implies \sum_{h \in f_a}\frac{1}{|f(h)|}\frac{n!}{\prod_r h_r!} (\prod_r \mu(\Gaw_r)^{h_r} - \prod_r \mu(\Gaw_r)^{h_{\sigma(r)}}) > 0$ for the permutation $\sigma$ switching the relative position of $a$ and $b$ in each report.

    This summation being positive implies the existence of an $h$ such that $f(h) = \{a\}$ and $\prod_r \mu(\Gaw_r)^{h_r} - \prod_r \mu(\Gaw_{\sigma(r)})^{h_r} > 0$.
    Moreover, $\inprod h {\R}_a > \langle h, \R \rangle_b$ since $f(h) = \{a\}$ and $f$ is a positional scoring rule.
    This implies the existence of a $t \in \R$ so that $\mu(\Gaw_t)^{h_t} - \mu(\Gaw_{\sigma(t)})^{h_t} > 0$ and $t_a > t_b$.The first statement implies $\mu(\Gaw_t) > \mu(\Gaw_{\sigma(t)})$, and the second implies $\sigma(t)_b > \sigma(t)_a$.
    Therefore,
    \begin{align*}
        \langle w, \sigma(t)-t \rangle &= w_a (\sigma(t)_a - t_a) + w_b (\sigma(t)_b - t_b) \\
        &= (w_b-w_a) (\sigma(t)_b - \sigma(t)_a) \geq 0~,
    \end{align*}
    which yields the contrapositive.
\end{proofE}

\begin{proposition}\label{prop:eps-uniform-sat-inc-SW}
    Suppose $f$ is an anonymous positional scoring rule, and $w$ and $v$ have the same order.
    Moreover, suppose $\mu$ is $\epsilon$-uniform for $\epsilon \leq \min_{s,t \in \R} \frac{|m(\Gaw_s) - m(\Gaw_t)|}{2}$ and $\frac{\sum_{k=1}^n {n \choose k}(R \epsilon)^k}{m}\leq \min_{a,b} \frac{|\nu(f_a) - \nu(f_b)|}{2}$, then expected social welfare increases.
\end{proposition}
\begin{proof}
    First, Lemma~\ref{lem:mu-same-order-wR} implies $\{\mu(\Gaw_{\phi(r)})\}$ has the same order as $\inprod w \R$, and therefore $\{\nusoc(f_a) - \nu(f_a)\}$ has the same ordering as $w$ by Lemma~\ref{lem:wR-same-order-nusocdiff}, which yields the conditions of Theorem~\ref{thm:inc-SW}.
\end{proof}

\begin{remark}
    While we show that expected social welfare increases if $w$ roughly aligns with social welfare, an adversarial anchoring point $w$ that is ordered in the reverse of $v$ will generally lead to a decrease in social welfare.
    This demonstrates the influence exhibited by the anchoring information $w$.
\end{remark}

Upon learning that welfare increases in expectation, a natural next question is to understand the probability the social welfare decreases.
Since political voting mechanisms are often run once every few years, understanding the probability that social welfare decreases is important, even when $w$ is perfectly aligned with expected utility $v$.

\begin{theoremE}[][]\label{thm:bound-prob-dec}
    Given a density $\mu$ such that preferences are drawn i.i.d. from $\mu$, and suppose the anchoring point $w$ and weight $\alpha \in [0,1]$.
    Then $\Pr_{\bld u}[\sw(f(\#\GammaR(\bld u)), \bld u) > \sw(f(\#\GammaR((1-\alpha)\bld u + \alpha w), \bld u)] \leq \E[\exp(\sw(f(\#\GammaR(\bld u)), \bld u) - \sw(f(\#\GammaR((1-\alpha)\bld u + \alpha w)), \bld u))]$.
    That is, the probability that social welfare decreases under anchored preferences is upper bounded.
\end{theoremE}
\begin{proof}
For notational simplicity, fix any $\epsilon \in (0, \min_{a, b : w_a > w_b} (w_a - w_b))$ and let $V := -\Delta \sw = \sw((1-\alpha)\bld u + \alpha w, \GammaR) - \sw(\bld u, \GammaR)$. 
\begin{align*}
\Pr[\Delta \sw < 0] &= \Pr[-\Delta \sw \geq \epsilon] &\text{Finite alternatives $m$}\\
        &\leq M_V(t) e^{-t \epsilon} \qquad \forall t > 0 &\text{Chernoff bound}\\
    &= \E_V [\exp (tV)] \exp(-t\epsilon) \qquad \forall t > 0 &\text{definition of MGF} \\
    &\leq \E_V [\exp(tV)] \qquad \forall t > 0 &\text{$-t\epsilon < 0 \implies \exp(-t\epsilon) \in (0,1)$}\\
    &\leq \E_V [\exp(V)] &\text{choose $t = 1$}\\
    &= \E_{\bld u} [\exp(\sw(f(\#\GammaR(\bld u), \bld u)) - \sw(f(\#\Gaw(\bld u)), \bld u)]
\end{align*}
\end{proof}
This bound is only meaningful (i.e., it is tighter than $1$) if $\E [\Delta \sw \mid w,\alpha] > 0$; that is, if voting by anchored preferences increases social welfare in expectation.
Moreover, as $\E [\Delta \sw \mid w,\alpha] \to \infty$ (i.e. social welfare increases by a significantly large quantity), the probability that social welfare decreases tends to $0$.

\section{Conclusions and future work}

This work provides methodologies to examine the outcomes of voting mechanisms \emph{in expectation} when voters are influenced by external information.
We often focus on the case where $w$ is aligned with social welfare. Unsurprisingly, imposing an anchoring point aligned with social welfare tends to increase social welfare in expectation.
Moreover, we give upper and lower bounds on the probability that any one alternative is selected by the voting mechanism, and show that, for the alternative $a$ maximizing $w_a$, this lower bound tightens, and the upper bound can loosen when anchored preferences.

Many directions for future work lie ahead.
For example, voters often receive external information conveying varying messages, and it would be fruitful to understand the effect of external information when there are multiple, conflicting sources of information.
Furthermore, it would be interesting to use these tools to examine how other mechanisms respond to socially-- or fairness--concerned players; do mechanisms with welfare or fairness guarantees lose those guarantees?

\newpage 

\section*{Acknowledgements}

We would like to thank Bailey Flanigan for helpful comments in scoping out and situating this work. 
This material is based upon work supported by the National Science Foundation under Award No. 2202898 and by National Science Foundation and Amazon under Award No. 2147187.

\bibliographystyle{named}
\bibliography{refs}

\begin{thebibliography}{}

\bibitem[\protect\citeauthoryear{Anshelevich \bgroup \em et al.\egroup
  }{2018}]{anshelevich2018approximating}
Elliot Anshelevich, Onkar Bhardwaj, Edith Elkind, John Postl, and Piotr
  Skowron.
\newblock Approximating optimal social choice under metric preferences.
\newblock {\em Artificial Intelligence}, 264:27--51, 2018.

\bibitem[\protect\citeauthoryear{Apesteguia \bgroup \em et al.\egroup
  }{2011}]{apesteguia2011justice}
Jose Apesteguia, Miguel~A Ballester, and Rosa Ferrer.
\newblock On the justice of decision rules.
\newblock {\em The Review of Economic Studies}, 78(1):1--16, 2011.

\bibitem[\protect\citeauthoryear{Arrow}{1951}]{arrow1951social}
Kenneth~J Arrow.
\newblock {\em Social choice and individual values}, volume~12.
\newblock Yale university press, 1951.

\bibitem[\protect\citeauthoryear{Bedaywi \bgroup \em et al.\egroup
  }{2023}]{bedaywi2023distortion}
Mark Bedaywi, Bailey Flanigan, Mohamad Latifian, and Nisarg Shah.
\newblock The distortion of public-spirited participatory budgeting.
\newblock 2023.

\bibitem[\protect\citeauthoryear{Boutilier \bgroup \em et al.\egroup
  }{2012}]{boutilier2012optimal}
Craig Boutilier, Ioannis Caragiannis, Simi Haber, Tyler Lu, Ariel~D. Procaccia,
  and Or~Sheffet.
\newblock Optimal social choice functions: A utilitarian view.
\newblock In {\em Proceedings of the 13th ACM Conference on Electronic
  Commerce}, EC '12, page 197–214, New York, NY, USA, 2012. Association for
  Computing Machinery.

\bibitem[\protect\citeauthoryear{Caragiannis \bgroup \em et al.\egroup
  }{2017}]{caragiannis2017subset}
Ioannis Caragiannis, Swaprava Nath, Ariel~D Procaccia, and Nisarg Shah.
\newblock Subset selection via implicit utilitarian voting.
\newblock {\em Journal of Artificial Intelligence Research}, 58:123--152, 2017.

\bibitem[\protect\citeauthoryear{Charikar \bgroup \em et al.\egroup
  }{2023}]{charikar2023breaking}
Moses Charikar, Prasanna Ramakrishnan, Kangning Wang, and Hongxun Wu.
\newblock Breaking the metric voting distortion barrier.
\newblock {\em arXiv preprint arXiv:2306.17838}, 2023.

\bibitem[\protect\citeauthoryear{Chen \bgroup \em et al.\egroup
  }{2014}]{chen2014altruism}
Po-An Chen, Bart~De Keijzer, David Kempe, and Guido Sch{\"a}fer.
\newblock Altruism and its impact on the price of anarchy.
\newblock {\em ACM Transactions on Economics and Computation (TEAC)},
  2(4):1--45, 2014.

\bibitem[\protect\citeauthoryear{Flanigan \bgroup \em et al.\egroup
  }{2023}]{flanigan2023distortion}
Bailey Flanigan, Ariel~D Procaccia, and Sven Wang.
\newblock Distortion under public-spirited voting.
\newblock {\em ACM Conference on Economics and Computation}, 2023.

\bibitem[\protect\citeauthoryear{Gibbard}{1973}]{gibbard1973manipulation}
Allan Gibbard.
\newblock Manipulation of voting schemes: a general result.
\newblock {\em Econometrica: journal of the Econometric Society}, pages
  587--601, 1973.

\bibitem[\protect\citeauthoryear{Lambert and
  Shoham}{2009}]{lambert2009eliciting}
Nicolas Lambert and Yoav Shoham.
\newblock Eliciting truthful answers to multiple-choice questions.
\newblock In {\em Proceedings of the 10th ACM conference on Electronic
  commerce}, pages 109--118, 2009.

\bibitem[\protect\citeauthoryear{Lambert}{2018}]{lambert2018elicitation}
Nicolas~S Lambert.
\newblock Elicitation and evaluation of statistical forecasts.
\newblock 2018.

\bibitem[\protect\citeauthoryear{Ledyard}{1997}]{ledyard1997public}
John Ledyard.
\newblock Public goods: A survey of experimental research, 1997.

\bibitem[\protect\citeauthoryear{Lehtinen and
  Kuorikoski}{2007}]{lehtinen2007unrealistic}
Aki Lehtinen and Jaakko Kuorikoski.
\newblock Unrealistic assumptions in rational choice theory.
\newblock {\em Philosophy of the Social Sciences}, 37(2):115--138, 2007.

\bibitem[\protect\citeauthoryear{Mitzenmacher \bgroup \em et al.\egroup
  }{2001}]{mitzenmacher2001power}
Michael Mitzenmacher, Andr\'ea Richa, and Ramesh Sitaraman.
\newblock The power of two random choices: A survey of techniques and results.
\newblock 2001.

\bibitem[\protect\citeauthoryear{Pennock and Chapman}{1977}]{pennock1977due}
James~Roland Pennock and John~W Chapman.
\newblock {\em Due Process: Nomos XVIII}, volume~30.
\newblock NYU Press, 1977.

\bibitem[\protect\citeauthoryear{Satterthwaite}{1975}]{satterthwaite1975strategy}
Mark~Allen Satterthwaite.
\newblock Strategy-proofness and arrow's conditions: Existence and
  correspondence theorems for voting procedures and social welfare functions.
\newblock {\em Journal of economic theory}, 10(2):187--217, 1975.

\bibitem[\protect\citeauthoryear{Schulte-Geers and
  Waggoner}{2022}]{schultegeers2022balls}
Ernst Schulte-Geers and Bo~Waggoner.
\newblock Balls and bins -- simple concentration bounds, 2022.

\bibitem[\protect\citeauthoryear{Skowron and Elkind}{2017}]{skowron2017social}
Piotr Skowron and Edith Elkind.
\newblock Social choice under metric preferences: Scoring rules and stv.
\newblock In {\em Proceedings of the AAAI Conference on Artificial
  Intelligence}, volume~31, 2017.

\bibitem[\protect\citeauthoryear{Staff}{2022}]{frenchstrategicvoting}
The Canadian~Press Staff.
\newblock France's election: `strategic' voting among montreal's french
  citizens, 2022.

\bibitem[\protect\citeauthoryear{Tsetlin \bgroup \em et al.\egroup
  }{2003}]{tsetlin2003impartial}
Ilia Tsetlin, Michel Regenwetter, and Bernard Grofman.
\newblock The impartial culture maximizes the probability of majority cycles.
\newblock {\em Social Choice and Welfare}, 21:387--398, 2003.

\end{thebibliography}

\newpage
\onecolumn
\appendix

\section{Omitted proofs}\label{app:omitted-proofs}

\printProofs

\end{document}